\documentclass[a4paper,USenglish,cleveref, autoref, thm-restate]{lipics-v2021}

\nolinenumbers

\author{Will Rosenbaum}{Amherst College}{wrosenbaum@amherst.edu}{https://orcid.org/0000-0002-7723-9090}{}

\authorrunning{W. Rosenbaum} 

\Copyright{William (Will) Rosenbaum} 



\category{} 

\relatedversion{} 

\acknowledgements{This work was born from numerous discussions with Boaz Patt-Shamir, to whom I am eternally grateful. I thank the anonymous reviewers for their thoughtful commentary which helped improve this paper.}

\EventEditors{Christian Scheideler}
\EventNoEds{1}
\EventLongTitle{36th International Symposium on Distributed Computing (DISC 2022)}
\EventShortTitle{DISC 2022}
\EventAcronym{DISC}
\EventYear{2022}
\EventDate{October 25--27, 2022}
\EventLocation{Augusta, Georgia, USA}
\EventLogo{}
\SeriesVolume{246}
\ArticleNo{4}

\usepackage{amsmath}
\usepackage{amsthm}
\usepackage{amsfonts}
\usepackage{algorithm}
\usepackage{algorithmic}
\usepackage[T1]{fontenc}

\newtheoremstyle{mythm}
     {3pt}
     {3pt}
     {}
     {}
     {\bfseries}
     {.}
     { }
     {}

\theoremstyle{mythm}
\newtheorem{thm}{Theorem}
\newtheorem{prop}[thm]{Proposition}
\newtheorem{lem}{Lemma}[section]
\newtheorem{cor}[lem]{Corollary}
\newtheorem{obs}[lem]{Observation}
\newtheorem{rem}[lem]{Remark}

\newtheorem{eg}[lem]{Example}
\newtheorem{dfn}[lem]{Definition}

\renewcommand{\descriptionlabel}[1]{\hspace\labelsep\normalfont\itshape#1}

\newcommand{\dft}[1]{\textbf{\textit{#1}}}

\newcommand{\abs}[1]{\left|#1\right|}

\newcommand{\floor}[1]{\left\lfloor#1\right\rfloor}

\newcommand{\paren}[1]{\left(#1\right)}
\newcommand{\set}[1]{\left\{#1\right\}}
\newcommand{\sucht}{\ \middle|\ }
\renewcommand{\th}{{}^{\mathrm{th}}}
\renewcommand{\st}{{}^{\mathrm{st}}}

\newcommand{\N}{\mathbf{N}}
\newcommand{\R}{\mathbf{R}}

\newcommand{\calA}{\mathcal{A}}
\newcommand{\calL}{\mathcal{L}}

\newcommand{\Arand}{A_{\mathrm{rand}}}

\DeclareMathOperator{\height}{ht}
\DeclareMathOperator{\ini}{ini}

\title{Packet Forwarding with a Locally Bursty Adversary}

\begin{document}

  \maketitle

  \begin{abstract}
    We consider packet forwarding in the adversarial queueing theory (AQT) model introduced by Borodin et al. We introduce a refinement of the AQT $(\rho, \sigma)$-bounded adversary, which we call a \emph{locally bursty adversary} (LBA) that parameterizes injection patterns jointly by edge utilization and packet origin. For constant ($O(1)$) parameters, the LBA model is strictly more permissive than the $(\rho, \sigma)$ model. For example, there are injection patterns in the LBA model with constant parameters that can only be realized as $(\rho, \sigma)$-bounded injection patterns with $\rho + \sigma = \Omega(n)$ (where $n$ is the network size). We show that the LBA model (unlike the $(\rho, \sigma)$ model) is closed under packet bundling and discretization operations. Thus, the LBA model allows one to reduce the study of general (uniform) capacity networks and inhomogenous packet sizes to unit capacity networks with homogeneous packets.

    On the algorithmic side, we focus on information gathering networks---i.e., networks in which all packets share a common destination, and the union of packet routes forms a tree. We show that the Odd-Even Downhill (OED) forwarding protocol described independently by Dobrev et al.\ and Patt-Shamir and Rosenbaum achieves buffer space usage of $O(\log n)$ against all LBAs with constant parameters. OED is a local protocol, but we show that the upper bound is tight even when compared to centralized protocols. Our lower bound for the LBA model is in contrast to the $(\rho, \sigma)$-model, where centralized protocols can achieve worst-case buffer space usage $O(1)$ for $\rho, \sigma = O(1)$, while the $O(\log n)$ upper bound for OED is optimal only for local protocols.
  \end{abstract}

  \section{Introduction}
\label{sec:intro}

Routing and forwarding are fundamental operations in the study of networks. In this context, commodities---for example, data packets, fluid flows, or physical objects---appear at various places in a network, and must be transferred to prescribed destinations. Movement is restricted by the network's topology. The goal is to get the commodities from source to destination as efficiently as possible. \emph{Routing} is the process of determining routes for the commodities to follow from source to destination, while \emph{forwarding} determines the particular schedule by which items---which we will henceforth refer to as \emph{packets}---move in the network. In this work, we focus on the process of forwarding packets, assuming their routes are pre-determined.

Two well-studied models of packet forwarding in networks are the adversarial queueing theory (AQT) model introduced by Borodin et al.~\cite{Borodin2001} and the network calculus model introduced by Cruz~\cite{Cruz1991-i, Cruz1991-ii}. In both models, packets are assumed to have prescribed routes from source to destination. Both models also parameterize packet arrivals in terms of long-term average rates and short-term burstiness in order to disallow trivially infeasible injection patterns that exceed network capacity constraints. AQT and network calculus also differ in some crucial ways. AQT examines injections of discrete, indivisible packets at discrete time intervals, and forwarding occurs in synchronous rounds. In network calculus, on the other hand, packets are modeled as continuous flows and forwarding is a continuous-time processes. Nonetheless, these flows can be discretized (or ``packetized'') to be processed discretely. One of the goals of this paper is to draw tighter connections between analogous parameters in the AQT and network calculus models under the process of discretization.

In both AQT and network calculus, one natural measure of efficiency is the buffer space usage of nodes in the network. That is, how much memory is required at each buffer in order to store packets that are \emph{en route} to their destinations. Traditionally, AQT has focused on a qualitative measure of space usage, called \emph{stability}, which merely requires that the space usage of a protocol remains bounded (by some function of the network parameters) for all time. A notable early exception is the work of Adler and Ros\'en~\cite{Adler2002}, which gives a quantitative buffer space upper bound for longest-in-system scheduling when the network is a directed acyclic graph.

AQT has also traditionally focused on \emph{greedy} forwarding protocols---i.e., protocols for which every non-empty buffer forwards as many packets as it can in each round (subject to capacity constraints). A more recent series of work~\cite{Miller2016, Dobrev2017-optimal, Patt-Shamir2017-space, Patt-Shamir2019-space-optimal, Miller2019-great} initiated by Miller and Patt-Shamir~\cite{Miller2016} studies quantitative buffer space bounds for non-greedy forwarding policies. In particular, these works show that in restricted network topologies (single-destination paths and trees), non-greedy forwarding protocols can achieve significantly better buffer space usage than greedy protocols. Specifically, non-greedy centralized forwarding protocols can achieve $O(1)$ buffer space usage~\cite{Miller2016, Miller2019-great}, while $\Theta(\log n)$ buffer space is necessary and sufficient for local (distributed) protocols~\cite{Dobrev2017-optimal, Patt-Shamir2017-space} (where $n$ is the number of buffers in the network). The work of Patt-Shamir and Rosenbaum~\cite{Patt-Shamir2019-space-optimal} shows there is a smooth trade-off between a protocol's \emph{locality} and optimal buffer space usage: if each node determines how may packets to forward based on the state of its distance $d$ neighborhood, then $\Theta(\frac 1 d \log n)$ buffer space is necessary and sufficient. These bounds are in contrast to greedy protocols, which require $\Omega(n)$ buffer space in the worst case.

The bounds described in the preceding paragraph refer to the AQT injection model in which edges in the network have uniform unit capacities---only one packet may cross any edge in a given round---and the average injection rate $\rho$ satisfies $\rho \leq 1$, and the burst parameter $\sigma$ satisfies $\sigma = O(1)$ (cf.~Definition~\ref{dfn:rho-sigma}). The algorithms can be generalized to general uniform edge capacities ($C \geq 1$ packets can cross each edge in a round), but the generalized algorithms are both more cumbersome to express, and correspondingly subtle to reason about (see, e.g., Section~1.1 in~\cite{Dobrev2017-optimal}). When dealing with general capacities, discretizations of continuous flows, and heterogeneous (indivisible) packets, a natural strategy is to bundle packets into ``jumbo packets''~\cite{Salyers2007-jumbogen}. This procedure can, however, lead to large bursts in the appearance of jumbo packets in the network, even if the injection process has a small burst parameter (see e.g., Remark~\ref{rem:discretization-burst}). Thus applying the \emph{analysis} of the relatively simple unit-capacity versions of algorithms in~\cite{Miller2016, Dobrev2017-optimal, Patt-Shamir2017-space, Patt-Shamir2019-space-optimal, Miller2019-great} directly to jumbo packets may not show any improvement over greedy algorithms.

\subsection{Our Contributions}

In this paper, we introduce a refinement of Borodin et al.~\cite{Borodin2001}'s $(\rho, \sigma)$ parameterization of injection patterns, that we call a the \dft{local burst model} (see Definition~\ref{dfn:local-burst}). In addition to the asymptotic rate $\rho$ and \emph{global} burst parameter $\sigma$, the local burst model has a third parameter, $\beta$---the \emph{local} burst parameter, that accounts for simultaneous bursts occurring at distinct injection sites. Thus, for small values of $\beta$, a locally $(\rho, \sigma, \beta)$-bounded injection pattern may still allow for large (e.g, $\Omega(n)$) simultaneous packet injections, so long as not too many packets are injected into the same buffer. 

We prove that locally bursty injection patterns are essentially characterized as discretizations of what we call ``locally dependent flows'' (Definition~\ref{dfn:dependent-flow}) with similar parameters---see Lemmas~\ref{lem:flow-relaxation} and~\ref{lem:discretization}. We use this characterization to show that applying packet bundling to a locally bursty injection pattern yields another locally bursty adversary with similar parameters---see Proposition~\ref{prop:uniform-bundling}. Consequently, any space efficient algorithm for unit capacity networks and homogeneous unit sized packets can be applied as a black-box to bundled jumbo packets to achieve similar buffer space usage to the unit capacity case (Corollary~\ref{cor:uniform-bundling}). We then show that a small modification of the framework can also be applied to the setting of heterogeneous packet sizes.

On the algorithmic side, we analyze the \emph{odd-even downhill} (OED) forwarding protocol of~\cite{Dobrev2017-optimal, Patt-Shamir2017-space} against locally bursty adversaries. We show that for constant ($O(1)$) parameters $\rho, \sigma$, and $\beta$, OED achieves worst-case buffer space $O(\log n)$ in information gathering networks of size $n$---see Theorem~\ref{thm:oed-ub}. This result is strictly stronger than the analyses of~\cite{Dobrev2017-optimal, Patt-Shamir2017-space}, as there are locally bursty injection patterns with $\beta = O(1)$ that can only be realized in the classical $(\rho, \sigma)$ model for $\sigma = \Omega(n)$. Combining our analysis of the OED protocol together with flow discretization and/or bundling, buffer space of $O(\log n)$ can be achieved for forwarding with general capacities, heterogeneous packets, and continuous flows---see Section~\ref{sec:ub-consequences}.

Finally, in Section~\ref{sec:buffer-size-lb}, we prove a matching lower bound of $\Omega(\log n)$ for any \emph{centralized randomized} protocol against locally bursty adversaries (Theorems~\ref{thm:det-lb} and~\ref{thm:rand-lb}). This lower bound is in contrast to the deterministic upper bounds of~\cite{Miller2016, Patt-Shamir2019-space-optimal, Miller2019-great} which show that $O(1)$ buffer space is achievable for $(\rho, \sigma)$-bounded adversaries for centralized and ``semi-local'' protocols. Thus, our lower bound shows that the local burst model (with constant parameters) gives the adversary strictly more power to inflict large ($\omega(1)$) buffer space usage on \emph{centralized} algorithms. The performance of the asymptotically optimal \emph{local} protocol is the same for the local burst and traditional injection models, and in the local burst model, the (local) OED protocol is asymptotically optimal, even when compared to centralized protocols.

\subsection{Discussion of Our Results}

The most natural application domains for this work are networks consisting of tightly synchronized nodes, such as network-on-chips (NoCs)~\cite{Kundu2018-network}, software defined networks (SDNs)~\cite{Schmid2013-exploiting}, and sensor networks. In these contexts, trees and grids are common network topologies. (In the case of grids, ``single bend'' routing allows one to treat the network essentially as a disjoint union of paths.) Thus, while the topologies we consider are highly restricted, the family of topologies is a fundamental and frequently used family for applications in which our techniques might be applied.

In NoCs, SDNs, and sensor networks, the rate and source of packet injections may be highly variable. Thus, the parameterizations of packet injections in both the standard AQT model and the network calculus model may be too coarse to model the actual buffer space requirement of observed injection patterns effectively. In the AQT model, allowing for multiple simultaneous packet injections into different buffers requires a large burst parameter, $\sigma$, even though the resulting injection pattern may be handled using buffer space $\ll\sigma$ (see Example~\ref{eg:directed-path}). On the other hand, the traditional network calculus model does not account for dependencies between rates of packet injections into different buffers over time. Thus, a standard analysis may severely overestimate the bandwidth or buffer space required to handle a given injection pattern. Our locally bursty injection model (and its continuous analogue described in Section~\ref{sec:flows}) refine both the AQT and network calculus models so as to give more precise bounds on the buffer space requirement of many natural packet injection patterns.

Together with the upper bounds of~\cite{Miller2016, Patt-Shamir2019-space-optimal, Miller2019-great}, Theorems~\ref{thm:det-lb} and~\ref{thm:rand-lb} imply that the locally bursty injection model (with constant parameters) gives an adversary strictly greater power to inflict large buffer space usage against \emph{centralized} and semi-local forwarding protocols. However, Theorem~\ref{thm:oed-ub} implies that the \emph{local} OED forwarding protocol achieves asymptotically optimal buffer space usage. Thus, for locally bursty injection patterns, there is no (asymptotic) advantage to implementing a centralized protocol, while OED still gives an exponential improvement over greedy protocols. We believe this insight may be valuable in VLSI design where protocols like OED could be implemented at a hardware level in order to reduce buffer space requirements. Hardware implementations of similar protocols have been proposed, for example in~\cite{Bund2020-pals}, in order to achieve decentralized (gradient) clock synchronization.









  
  \section{Model and Preliminaries}

We model a packet forwarding network as a directed graph, $G = (V, E)$. Each edge $e = (u, v) \in E$ has an associated buffer that stores packets in node $u$ as they wait to cross the edge $(u, v)$ to $v$. We use the notation $e = (u, v)$ to denote both the edge in $G$ and its associated buffer.

In our model, an execution proceeds in synchronous rounds. Each round consists of two steps: an \dft{injection step} in which new packets arrive in the network, and a \dft{forwarding step} in which buffers forward packets across edges of the graph. During the forwarding step, each buffer chooses a subset of packets to forward, and forwards those packets across the edge $(u, v)$ associated with the buffer. These packets arrive at their next location---either another buffer in node $u$, or are delivered to their destination---before the beginning of the next round. Each edge $e$, has a \dft{capacity} $C(e)$, which is the maximum number of packets that can cross $e$ in a single forwarding step. 

At a given time $t$, we use $L^t(e)$ to denote the contents of buffer $e$ during round $t$ between the injection and forwarding steps. $\abs{L^t(e)}$ is the \dft{load} of $e$---i.e., number of packets stored in buffer $e$.

A \dft{packet} $p$ is a pair $(t, P)$ where $t \in \N$ and $P = (v_0, v_1, \ldots, v_{\ell})$ is a directed path in $G$. The interpretation is that $t$ indicates the time (round) at which $P$ is injected, and $P$ specifies a \dft{route} from $P$'s \dft{source}, $e_0 = (v_0, v_1)$ to $P$'s \dft{destination} $v_\ell$. An \dft{adversary} or \dft{injection pattern} $A$ is a multi-set of packets.

Given a packet $p = (t, (v_0, v_1, \ldots, v_\ell))$, we say that $p$'s route \dft{contains} an edge $e \in E$ if $e = (v_i, v_{i+1})$ for some $i \in [\ell - 1]$. For a fixed adversary $A$ and time interval $T = [r, s] \subseteq \N$, we define $N^T(e)$ to be the number of packets injected during times $t \in T$ whose routes contain $e$. That is
\[
N^T(e) = \abs{\set{(t, P) \in A \sucht t \in T \text{ and } P \text{ contains } e}}.
\]
We also define a more refined measure of utilization of an edge $e$ that differentiates packets according to their origins. Specifically, for any subset $S \subseteq E$, we define
\[
N_S^T(e) = \abs{\set{(t, P) \in A \sucht t \in T, (v_0, v_1) \in S, \text{ and } P \text{ contains } e}}
\]
In particular, we have $N^T(e) = N_E^T(e)$. In the adversarial queueing model (AQT) of Borodin et al.~\cite{Borodin2001}, the edge utilization of an adversary $A$ is parameterized is follows.

\begin{dfn}
  \label{dfn:rho-sigma}
  Given $\rho > 0$ and $\sigma \geq 0$, we say that an adversary $A$ is \dft{$(\rho, \sigma)$-bounded} if for all $e$ and (finite) intervals $T \subseteq \N$ we have
  \begin{equation}
    \label{eqn:rho-sigma}
    N^T(e) \leq \rho \abs{T} + \sigma.    
  \end{equation}
  We denote the family of $(\rho, \sigma)$-bounded adversaries by $\calA(\rho, \sigma)$.
\end{dfn}

For a $(\rho, \sigma)$-bounded adversary, the parameter $\rho$ is an upper bound on the maximum average utilization of an edge in the network, while $\sigma$ measures ``burstiness''---the amount by which the average can be exceeded over any time interval. For example, taking $T$ with $\abs{T} = 1$,~(\ref{eqn:rho-sigma}) implies that at most $\rho + \sigma$ packets are injected into any buffer in any single round.

\subsection{Locally Bursty Adversaries}

Here, we define a more refined parameterization of adversaries, which we call the \dft{local burst model}. We refer to adversaries parameterized by the local burst model as \dft{locally bursty adversaries}, or \dft{LBA}s.



\begin{dfn}
  \label{dfn:local-burst}
  Let $A$ be an adversary, $\rho > 0$, $\sigma \geq 0$ and $\beta : E \to \N$. Then we say that $A$ is \dft{locally $(\rho, \sigma, \beta)$-bounded} if for all finite intervals $T \subseteq \N$, subsets $S \subseteq E$ and $e \in E$, we have
  \begin{equation}
    \label{eqn:local-burst}
    N_S^T(e) \leq \rho \abs{T} + \sigma +  \sum_{f \in S} \beta(f).    
  \end{equation}
  That is, for every subset $S$ of buffers, the rate of injections into $S$ that cross $e$ only ever exceeds $\rho$ by $\sigma$ more than the sum of the $\beta(e)$ for $e \in S$. We denote the family of local $(\rho, \sigma, \beta)$-bounded adversaries by $\calL(\rho, \sigma, \beta)$. In the case that $A \in \calL(\rho, \sigma, \beta)$ and there is a constant $B$ such that $\beta(e) \leq B$ for all buffers $e$, we will say that $A$ is local $(\rho, \sigma, B)$-bounded.
\end{dfn}

We formalize the following observation that gives a relationship between the parameters of $(\rho, \sigma)$-bounded adversaries and local $(\rho, \sigma, \beta)$-bounded adversaries.

\begin{obs}
  \label{obs:global-vs-local-burst}
  Fix a network $G$ and parameters $\rho$, $\sigma$, and $\beta : E \to \N$. Suppose $A \in \calL(\rho, \sigma, \beta)$. Then $A \in \calA(\rho, \sigma')$ for $\sigma' = \sigma + \sum_{e \in E} \beta(e)$.
\end{obs}

\begin{eg}
  \label{eg:directed-path}
  Let $G$ be the \dft{single-destination path} of size $n$. That is, $G = (V, E)$ where $V = \set{1, 2, \ldots, n, n+1}$ and $E = \set{(i, i+1) \sucht i \in [n]}$. Further, all injected packets have destination $n+1$. We consider two injection patterns, $A_0$ and $A_1$
  \begin{description}
  \item[$A_0$:] in rounds $1, n + 1, 2 n + 1, \ldots, k n + 1, \ldots$, there are $n$ packets injected into buffer $1$ with destination $n + 1$.
  \item[$A_1$:] in rounds $1, n + 1, 2 n + 1, \ldots, k n + 1, \ldots$, one packet is injected into each buffer $i = 1, 2, \ldots, n$ with destination $n+1$. 
  \end{description}
  Observe that both adversaries are in $\calA(1, n - 1)$, but not in $\calA(1, \sigma)$ for any $\sigma < n - 1$. Thus, the parameters of Definition~\ref{dfn:rho-sigma} do not distinguish $A_0$ and $A_1$. Yet $A_0$ and $A_1$ have vastly different buffer space requirements. $A_0$ requires buffer $1$ to have space $n$ for any forwarding protocol, while simple greedy forwarding for $A_1$ will achieve buffer space usage $\abs{L^t(i)} \leq 1$ for all $t$ and $i$.

  The parameters of the local burst model, however, can distinguish between $A_0$ and $A_1$. $A_1 \in \calL(1, 0, 1)$ (i.e., $\beta(i) = 1$ for all $i$), while $A_0 \in \calL(1, \sigma, \beta)$ only for $\sigma + \beta(1) \geq n - 1$. We will show that in the case of \emph{information gathering networks}---networks in which all packets share a common destination and the union of their routes forms a tree---all local $\calL(1, \sigma, B)$-bounded adversaries can be forwarded using $O(B \log n + \sigma)$ space. Thus, the \emph{local} burst parameter $\beta$ gives a more refined understanding of the buffer space requirement of a given injection pattern. 
\end{eg}

\subsection{Flows}
\label{sec:flows}

Another well-studied model for packet forwarding is the network calculus model introduced by Cruz~\cite{Cruz1991-i, Cruz1991-ii}. In the network calculus, packets are associated with \emph{flows}, and their arrivals are modeled as continuous time processes.

\begin{dfn}
  \label{dfn:flow}
  Given a network $G = (V, E)$, A \dft{flow} $\phi = (a, P)$ consists of a right-continuous arrival curve $a : \R \to \R$ and associated path $P$. We say that $\phi$ has \dft{rate} (at most) $r$ and burst parameter $b$ if for all $s < t$, the arrival curve $a$ satisfies
  \begin{equation}
    \label{eqn:flow-bound}
    a(t) - a(s) \leq r \cdot (t - s) + b.
  \end{equation}
  By convention, we assume $a(t) = 0$ for all $t < 0$.
\end{dfn}

For a single flow $\phi$, the parameters $r$ and $b$ are analogous to the rate and burst parameters $\rho$ and $\sigma$ in Definition~\ref{dfn:rho-sigma}. However, in a flow $\phi$, all packets share a common route, $P$. In particular, all packets associated with $\phi$ are injected to the same buffer and have the same destination.

In order to consider scenarios in which packets have multiple routes, we must consider multiple concurrent flows. In this setting, the analogy between equations~(\ref{eqn:rho-sigma}) and~(\ref{eqn:flow-bound}) breaks down, as the former bounds the total number packets utilizing any particular edge, while the latter bounds the arrivals of packets in flows (i.e., along entire paths, rather than individual edges). In order to tighten the connection between the AQT injection model and flows, we introduce a \dft{dependent} flow model in which we constrain the sum of arrival rates of flows across edge.

\begin{dfn}
  \label{dfn:dependent-flow}
  Let $G = (V, E)$ be a network and $\Phi = \set{\phi}$ be a family of flows. For an edge $e \in E$, let $\Phi_e$ denote the set of flows in $\Phi$ whose paths contain $e$. That is,
  \[
  \Phi_e = \set{(a, P) \in \Phi \sucht e \in P}.
  \]
  Suppose each $\phi \in \Phi$ obeys a $r_\phi, b_\phi$ bound as in~(\ref{eqn:flow-bound}). We say that $\Phi$ obeys a \dft{locally dependent rate bound} $r$ and \dft{global burst parameter} $\sigma$ if for every edge $e$, every set $\Psi \subseteq \Phi_e$ of flows, and all times $s, t$, we have
  \begin{equation}
    \label{eqn:dependent-flow}
    \sum_{\phi \in \Psi} (a_\phi(t) - a_\phi(s)) \leq r \cdot (t - s) + \sigma + \sum_{\phi \in \Psi} b_\phi.
  \end{equation}
\end{dfn}

We note the similarity between equations~(\ref{eqn:dependent-flow}) and~(\ref{eqn:local-burst}). In fact, Definition~\ref{dfn:dependent-flow} is a strict generalization of the LBA model: Given any injection pattern $A$, we can associate a family $\Phi_A$ of flows with $A$. Specifically, we define $\Phi_A$ to be
\begin{equation}
  \label{eqn:phi-a}
    \Phi_A = \set{(a, P) \sucht (t, P) \in A \text{ and } a(t) = \sum_{s \in \N,\ s \leq t} |\{(s, P) \in A\}|}
\end{equation}
With this association, the following lemma is clear.

\begin{lem}
  \label{lem:flow-relaxation}
  Suppose $A$ is a locally $(\rho, \sigma, \beta)$-bounded adversary, and let $\Phi_A$ be  the corresponding flow defined by~(\ref{eqn:phi-a}). Then for each flow $\phi = (a_\phi, P_\phi) \in \Phi_A$, $a$ has rate at most $\rho$, global burst parameter $\sigma$, and local burst parameter $b_\phi = \beta(\ini_\phi)$, where $\ini_\phi$ denotes the initial buffer in $\phi$'s path. Moreover, $\Phi_A$ obeys a locally dependent rate bound of $\rho$.
\end{lem}

Conversely, LBAs arise naturally as discretizations (packetizations) of (locally dependent) flows. We formalize this connection in the following definition and lemma.

\begin{dfn}
  \label{dfn:flow-discretization}
  Let $G = (V, E)$ be a network and $\Phi$ a family of flows on $G$. The \dft{discretization} of $\Phi$ is the AQT injection pattern $A_\Phi$ defined as follows. For each flow $\phi = (a_\phi, P_\phi) \in \Phi$ and time $t = \N$, $A_\Phi$ contains $\floor{a_\phi(t)} - \floor{a_\phi(t-1)}$ packets injected at time $t$ with route $P_\phi$.
\end{dfn}

We can view the discretization of a flow as forming packets via the following process. Each buffer maintains a set of (complete) packets, as well as a reserve of ``fractional'' packets associated with each flow originating at the buffer. At times $s \in (t-1, t]$, flows enter a buffer~$e$. At time $t$, the integral parts of each flow that has not yet been bundled as packets are injected as complete packets into the buffer, while the fractional remainder is reserved. The following lemma shows that for flows obeying a locally dependent rate bound, the resulting packet injection pattern is locally bounded as well.

\begin{lem}
  \label{lem:discretization}
  Let $G = (V, E)$ be a graph and $\Phi$ a family of flows on $G$. For each $\phi \in \Phi$, let $\ini_\phi$ denote the initial buffer in $\phi$'s path.  Suppose $\Phi$ obeys a locally dependent rate bound of $r$ with global burst parameter $\sigma$, and define the function $\beta : E \to \N$ by
  \[
  \beta(e) = \sum_{\phi\,:\, \ini_\phi = e} (1 + b_\phi).
  \]
  Then the discretization $A_\Phi$ of $\Phi$ is locally $(r, \sigma, \beta)$ bounded.
\end{lem}
\begin{proof}
  Fix a set $S \subseteq E$ of initial buffers, an edge $e$, and (discrete) time interval $T = [t_0, t_1]$. Let $A = A_\Phi$, and let $\Psi \subseteq \Phi_e$ be the subset of flows containing $e$ and with origin in $S$. We compute
  \begin{align}
    N_S^T(e) &= \abs{\set{(t, P) \in A \sucht t \in T, (v_0, v_1) \in S, \text{ and } P \text{ contains } e}}\nonumber\\
    &= \sum_{\phi \in \Psi} \sum_{s \in T} \paren{\floor{a_\phi(s)} - \floor{a_\phi(s-1)}}\nonumber\\
    &= \sum_{\phi \in \Psi} \paren{\floor{a_\phi(t_1)} - \floor{a_\phi(t_0 - 1)}}\nonumber\\
    &\leq r \cdot (t_1 - t_0 + 1) + \sigma + \sum_{\psi \in \Phi} (1 + b_\phi)\label{eqn:b-phi}\\
    &= r \cdot \abs{T} + \sigma + \sum_{e \in S} \beta(e)\nonumber.
  \end{align}
  In Equation~(\ref{eqn:b-phi}), we use the fact that $\floor a - \floor b \leq 1 + a - b$.
\end{proof}

\begin{rem}
  \label{rem:discretization-burst}
  The result of Lemma~\ref{lem:discretization} is a significant refinement of the analogous statement for the standard $(\rho, \sigma)$ burst model. To see this, consider the single destination path (Example~\ref{eg:directed-path}), and take $\Phi = \set{\phi_1, \phi_2, \ldots, \phi_n}$ to be the family of flows where each $\phi_i$ has arrival curve $a(t) = \frac 1 n t$ and associated path $P_i = (i, i+1, \ldots, n+1)$. In the associated injection pattern $A_\Phi$, one packet is injected into every buffer at times $n, 2 n, 3 n, \ldots$ (cf.\ $A_1$ in Example~\ref{eg:directed-path}). Even though flows in $\Phi$ have burst parameter $0$, large bursts appear in $A_\Phi$ as the result of the rounding process. Nonetheless, Lemma~\ref{lem:discretization} asserts that $A_\phi$ is \emph{locally} $(1, 0, 1)$-bounded, while the injection pattern is only $(1, \sigma)$-bounded for $\sigma \geq n - 1$. 
\end{rem}

  \section{Packet Bundling}
\label{sec:bundling}

In this section we assume that for a network $G = (V, E)$, all edges have the same (integral) capacity $C$. We examine the following strategy for dealing with general uniform capacity networks: when packets arrive in a buffer, they are set in a reserve buffer until sufficiently many (e.g., $C$) packets occupy the reserve buffer. Then the packets are bundled together, and treated as one indivisible ``jumbo'' packet. This process is appealing because if all jumbo packets have size $C$, then forwarding protocols designed for unit capacities can be applied to jumbo packets. Thus, the approach sidesteps potential subtleties in reasoning about general capacities (see Section~1.1 in~\cite{Dobrev2017-optimal}).

Our main results in this section show that if the original packet injection pattern obeys an LBA bound, then the resulting injection pattern of jumbo packets obeys a similar LBA bound with the parameters scaled down. Thus, if any algorithm guarantees some buffer space usage for unit capacity networks, then applying the same algorithm to jumbo packets will automatically give an analogous bound for general capacities.

\subsection{Uniform Packets}
\label{sec:uniform-packets}

We first consider the case where all packets have unit size (as in the standard AQT model), but all edges in the network have capacity $C \geq 1$. Now let $A$ be any locally $(C, \sigma, \beta)$-bounded adversary, and let $\Phi = \Phi_A$ be the corresponding family of flows. We define the \dft{$C$-reduction} of $\Phi$, denoted $\frac{1}{C} \Phi$, to be
\[
\frac 1 C \Phi = \set{\paren{\frac 1 C a, P} \sucht (a, P) \in \Phi}.
\]
Similarly, we define the $C$-reduction of $A$, denoted $\frac{1}{C} A$, to be the discretization (Definition~\ref{dfn:flow-discretization}) of $\frac 1 C \Phi$.

Observe that $\frac{1}{C} A$ is derived from $A$ via precisely the process of forming jumbo packets as described above. The following proposition follows immediately from Lemmas~\ref{lem:flow-relaxation} and~\ref{lem:discretization}.

\begin{prop}
  \label{prop:uniform-bundling}
  Suppose $A$ is locally $(\rho, \sigma, \beta)$-bounded. Then, $\frac 1 C A$ is locally $(\rho / C, \sigma / C, 1 + \beta / C)$-bounded.
\end{prop}

Again, we emphasize that the analogue of Proposition~\ref{prop:uniform-bundling} is not true for the standard $(\rho, \sigma)$-bounded adversary model. The proposition has the following consequence.

\begin{cor}
  \label{cor:uniform-bundling}
  Suppose $F$ is a forwarding protocol that for any locally $(1, \sigma, \beta)$-bounded adversary $A$ on a unit-capacity network $G = (V, E)$ achieves buffer space usage
  \[
  \sup_{e \in E, t \in \N} \abs{L^t(e)} \leq f_G(\sigma, \beta).
  \]
  Then for any uniform capacity $C$ and locally $(C, \sigma, \beta)$-bounded adversary $A$, applying $F$ to $\frac 1 C A$ achieves buffer space usage
  \[
  \sup_{e \in E, t \in \N} \abs{L^t(e)} \leq C f_G(\sigma / C, 1 + \beta / C) + C.
  \]
\end{cor}

We note that the additive $C$ term in the final expression comes from the need to store packets that have not yet been bundled.

\subsection{Heterogeneous Packets}
\label{sec:heterogeneous}

The framework described in Sections~\ref{sec:flows} and~\ref{sec:uniform-packets} shows how forwarding protocols for the AQT model with unit edge capacities can be applied to (1) discretizations of continuous flows, and (2) AQT adversaries with arbitrary uniform edge capacities and (uniform) unit-sized packets. Here, we describe a slight modification of the framework that allows for indivisible packets with heterogeneous sizes. To this end, we augment the AQT model as follows:
\begin{itemize}
\item Each packet $p$ has an associated size, denoted $w(p)$.
\item In a single round, an edge with capacity $C$ can forward a set of packets whose sizes sum to at most $C$.
\item An adversary $A$ is locally $(\rho, \sigma, \beta)$ bounded if for any subset $S$ of buffers, any edge $e$, and in any $T$ consecutive rounds, the sum of \emph{sizes} of packets injected into $S$ whose paths contain $e$ is at most $\rho \cdot T + \sigma + \sum_{f \in S} \beta(f)$.
\end{itemize}

The following example shows one complication caused by indivisible heterogeneous packets.

\begin{eg}
  \label{eg:heterogeneous}
  Consider a single edge $e$ with capacity $1$. Then a $(1,1)$-bounded adversary can inject $3$ packets of size $2/3$ every $2$ rounds that must cross $e$. Since $e$ has capacity $1$, it can only forward a single packet each round. Thus, the injection pattern is infeasible (i.e., cannot be handled with finite buffer space).
\end{eg}

We can preclude infeasible injection patterns (such as Example~\ref{eg:heterogeneous}) by further restricting the allowable injection rate. Consider the following bundling procedure: when packets are injected into a buffer, they are placed in a reserve buffer. If the load of the reserve buffer exceeds $\frac{1}{2} C$, then its contents are bundled into packets, each of whose total load is at least $\frac{1}{2} C$. Arguing as before, if the original injection pattern is locally $(\rho, \sigma, \beta)$-bounded with $\rho \leq \frac 1 2 C$, then the resulting injection pattern of bundled packets is locally $(1, \sigma / C, 1 + \beta / C)$-bounded. Note that even though the rate of the adversary is $\rho \leq 1/2$, the rate of the bundled injections can be as large as $1$. This occurs, for example, if the adversary always injects $C/2$ packets into a single buffer each round. These packets are then bundled, resulting in one complete bundle appearing each round.

  \section{OED Upper Bound}
\label{buffer-size-ub}

In this section we prove worst-case buffer space upper bounds for ``information gathering networks'' with a locally bursty adversary---i.e., instances in which all packets share a common destination and the union of trajectories of all packets forms a tree. Specifically, we show that the odd-even downhill (ODE) algorithm of~\cite{Dobrev2017-optimal, Patt-Shamir2017-space} requires $O(B \log n + \sigma)$ buffer space for any $(\rho, \sigma, \beta)$-bounded adversary for which $\beta(e) \leq B$ for all buffers $e$. For clarity and notational simplicity we describe the algorithm and argument in the simpler setting where the network consists of a path. All of the results remain true for general information gathering networks, and analogous arguments follow using the terminology and preliminary results described in~\cite{Patt-Shamir2017-space}.

In the case of the single destination path, the network $G = (V, E)$ consists of a path: $V = \set{1, 2, \ldots, n+1}$, and $E = \set{(i, i+1) \sucht i \leq n}$. All packets share the destination $n+1$, though they can be injected into any buffer. To cut down on notational clutter, we associate a buffer $(i, i+1)$ with its index $i$. In this setting, we describe the \dft{Odd-Even Downhill} or \dft{OED} algorithm independently introduced by Dobrev et al.~\cite{Dobrev2017-optimal} and  Patt-Shamir and Rosenbaum~\cite{Patt-Shamir2017-space}. Following Section~\ref{sec:bundling}, we assume that all edge capacities are $1$, and that all packets have unit size. To simplify notation, we use $L^t(i)$ to denote the the \emph{number} of packets in buffer $i$ immediately before the forwarding step of round $t$.

\begin{dfn}
  The \dft{OED rule} stipulates that $i$ forwards a packet in round $t$ if and only if one of the following conditions is satisfied:
  \begin{enumerate}
  \item $L^t(i) > L^t(i+1)$, or
  \item $L^t(i) = L^t(i+1)$ and $L^t(i)$ is \emph{odd}.
  \end{enumerate}
  By convention, we set $L^t(n+1) = 0$ for all $t$.
\end{dfn}

Both original papers~\cite{Dobrev2017-optimal, Patt-Shamir2017-space} show that for all $(\rho, \sigma)$-bounded adversaries, the maximum buffer load under OED forwarding is $O(\log n + \sigma)$. We will show that OED forwarding achieves similar buffer space usage for any local $(\rho, \sigma, \beta)$-bounded adversary.

\begin{thm}
  \label{thm:oed-ub}
  Let $G$ be a single-destination path of size $n+1$, and let $A$ be any local $(\rho, \sigma, \beta)$-bounded adversary where $\beta(i) \leq B$ for all $i$. Then the worst case buffer load is $O(B \log n + \sigma)$. That is,
  \[
  \sup_{t, i} L^t(i) = O(B \log n + \sigma).
  \]
\end{thm}

Our proof of Theorem~\ref{thm:oed-ub} follows the analysis of OED presented in~\cite{Patt-Shamir2017-space}. Specifically, their analysis considers the evolution of \emph{plateaus} in the network.

\begin{dfn}[cf.~\cite{Patt-Shamir2017-space}]
  \label{dfn:plateau}
  Let $G$ be a single destination path and $L^t : V \to \N$ a \dft{configuration}---i.e., assignment of loads to buffers---at some time $t$. We say that an interval $I = [a, b] \subseteq V$ is a \dft{plateau} of \dft{height} $h$ if $I$ is a maximal sub-interval of $V$ such that for all $i \in I$, $L^t(i) \geq h$. That is, every buffer in $I$ has load at least $h$, and there is no larger interval $I'$ containing $I$ with this property. We say that $I$ is an \dft{even plateau} if $I$ is a plateau of height $h$ for some even number $h$.
\end{dfn}

\begin{figure}
  \includegraphics[scale=0.75]{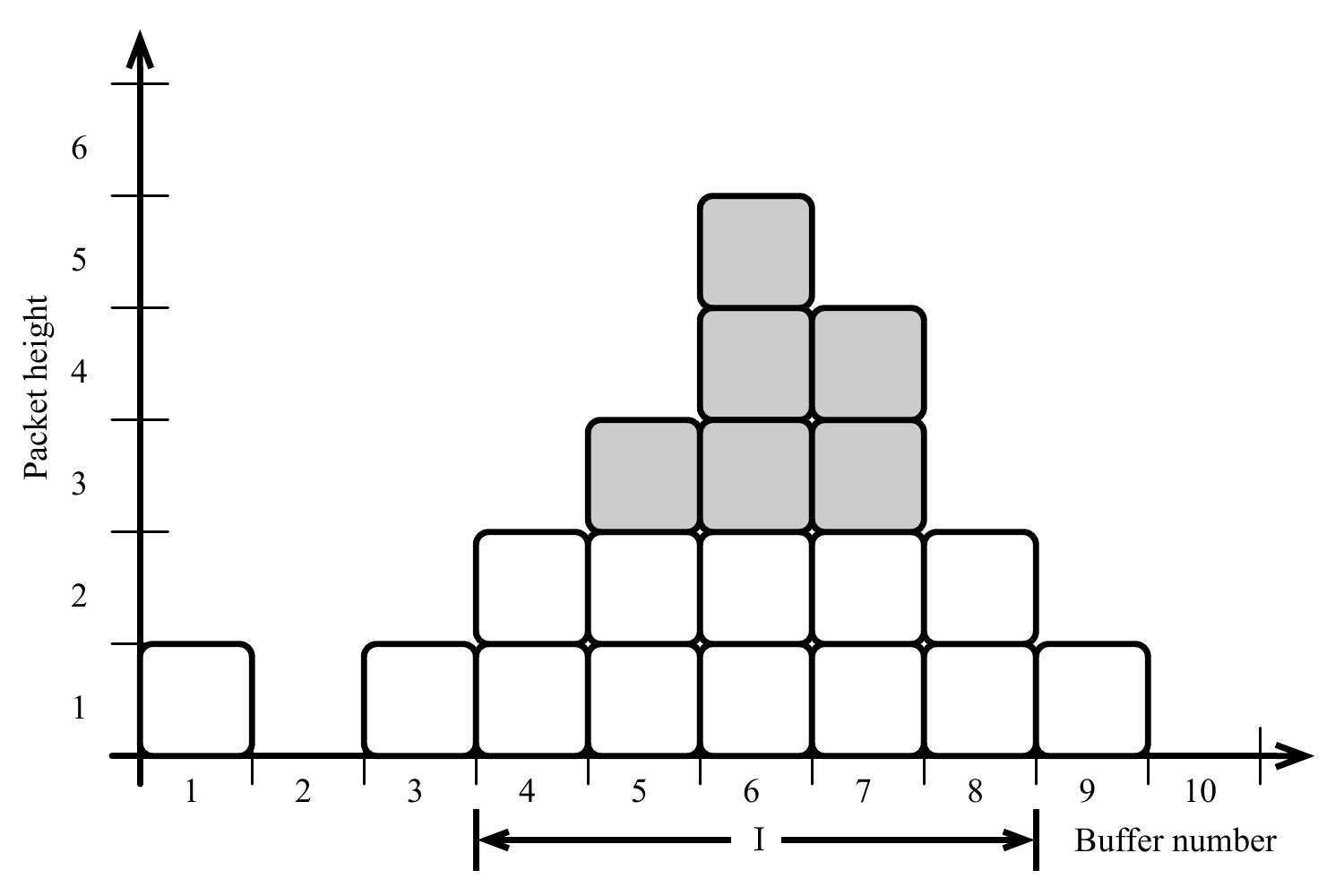}
  \caption{A configuration of packets. Each column represents a buffer, labeled $1$ through $10$, while the vertical axis indicates heights. The load of each buffer corresponds to the height of the highest packet in the buffer; for example, $L(6) = 5$. The indicated interval $I = [4, 8]$ is a plateau of height $2$. The six shaded packets sit above the plateau $I$, so that $L_2(I) = 6$. Corollary~\ref{cor:upper-load} states that the number of packets above $I$ is at most $(B+1) \abs{I} + \sigma$ for any locally $(1, \sigma, B)$-bounded adversary. The proof of Theorem~\ref{thm:oed-ub} applies this corollary inductively to the nested sequence of even plateaus containing the buffer with maximum load to show that the maximum load is at most logarithmic in the size of the network.\label{fig:plateau}}  
\end{figure}

We think of packets in $G$ as being arranged vertically in buffers---see Figure~\ref{fig:plateau}. Since packets share the same destination, for the purposes of our load analysis, we can treat all packets in a buffer as indistinguishable.\footnote{In order to analyze packet latency, one should distinguish packets by their age.} We refer to the \dft{height} of a packet in a buffer as one greater than the number of packets below it. Thus, for a buffer with load $1$, its sole packet is at height $1$; a buffer with two packets has one at height $1$ and the second at height $2$, etc. We say that a packet $P$ is \dft{above} a plateau $I$ of height $h$ if $\height(P) > h$. Given a configuration $L : V \to \N$ and a plateau $I$ of height $h$, we denote the number of packets above $I$ by $L_h(I)$. That is,
\begin{equation}
  \label{eqn:h-load}
  L_h(I) = \sum_{i \in I} (L(i) - h).
\end{equation}

OED forwarding does not specify \emph{which} packet is forwarded when a buffer is forwarded, and the maximum load analysis of the algorithm is independent of this choice. Nonetheless, for the purposes of bookkeeping, it will be convenient to adopt the following conventions:
\begin{enumerate}
\item Whenever a packet is injected or received as the result of forwarding, it occupies the highest position in its buffer;
\item When a buffer forwards a packet, the highest packet in the buffer is forwarded.
\end{enumerate}
That is---for the purposes of bookkeeping---we assume that the buffers operate as LIFO (last-in, first-out) stacks.

With these conventions, we make some preliminary observations about the movement of packets in an execution of the OED algorithm. 

\begin{lem}
  \label{lem:persistence}
  Suppose $L : V \to \N$ is a configuration immediately before forwarding and $L'$ the configuration afterward. Suppose $I = [a, b] \subseteq V$ is an even plateau of height $h$ in configuration $L$. Then in configuration $L'$, for all $i \in [a, b-1]$, we have $L'(i) \geq h$. Thus, in $L'$, the interval $[a, b-1]$ is contained in a plateau $I'$ of height $h$.
\end{lem}
\begin{proof}
  Suppose $i \in [a, b - 1]$. Then in $L$, we have $L(i), L(i+1) \geq h$. Since $h$ is even, $i$ will not forward unless $L(i) \geq h + 1$. Therefore, $L'(i) \geq h + 1 - 1 = h$.
\end{proof}

\begin{lem}
  \label{lem:packet-movement}
  For any packet $P$, let $\height^t(P)$ denote the height of $P$ at time $t$,\footnote{Note that this quantity is well-defined for all $t$ (until $P$ is delivered to its destination) by the LIFO conventions for height and packet movement.} and let $h$ be an even number. If $\height^t(P) \leq h$, then for all $s \geq t$ we have $\height^s(P) \leq h$.
\end{lem}
\begin{proof}
  By our conventions of packet movement, the height of $P$ in a fixed buffer is unchanged until it is forwarded. Since the top packet is always forwarded, $P$ can only be forwarded if $\height(P) = L(i)$, where $i$ is the buffer containing $P$ before forwarding. Since $\height(P) \leq h$ we also have $L(i) \leq h$. Since $h$ is even, $i$ only forwards if $L(i+1) \leq h - 1$. Therefore, the height of $P$ \emph{after} forwarding is at most $L(i+1) + 1 \leq h - 1 + 1 = h$. Thus, $P$ remains at height at most $h$.
\end{proof}

\begin{cor}
  \label{cor:even-plateau}
  Suppose $L : V \to \N$ is a configuration and $I = [a, b]$ is an even plateau of height $h$. Suppose a packet $P$ was injected at time $s$, and at time $t \geq s$ sits above $I$ (i.e., $P$ occupies a buffer $j \in I$ and $\height^t(P) > h$). Then $P$ was injected into a buffer in $I$ with an initial height $\height^s(P) > h$.
\end{cor}
\begin{proof}
  Suppose that at time $t$, $P$ occupies buffer $j \in [a, b]$. Since packets are only forwarded to a buffer with larger index, $P$ was injected at some buffer $i \leq j$ at some time $t_0 \leq t$. By Lemma~\ref{lem:packet-movement}, for all $s \in [t_0, t]$ we have $\height^{s}(P) > h$, whence the second assertion of the corollary holds.

  For the first assertion, we must show that $i \in I$. To this end, for each time $s \in [t_0, t]$, let $I_s = [a_s, b_s]$ denote the plateau of height $h$ in which $P$ is contained at time $s$. Thus, $I_{t_0}$ is the plateau above which $P$ is initially injected so that $j \in I_{t_0}$, and $I_{t} = I = [a, b]$. By Lemma~\ref{lem:persistence}, for all $s < t$, we have $[a_s, b_s - 1] \subseteq I_{s+1} = [a_{s+1}, b_{s+1}]$. Therefore, we have $a_s \geq a_{s+1}$ for all $s$. Thus, by induction, we have $a_t \leq a_{t_0} \leq j$, where the second inequality holds because $j \in I_{t_0}$. This gives the desired result.
\end{proof}

We now quote a lemma from~\cite{Patt-Shamir2017-space}, which bounds the number of packets above even plateaus for $(\rho, \sigma)$-bounded adversaries.

\begin{lem}[cf.~Lemma~3.4~in~\cite{Patt-Shamir2017-space}]
  \label{lem:upper-load}
  Let $A$ be a $(\rho, \sigma)$-bounded adversary and suppose $L$ is a configuration realized by OED forwarding. Suppose $I$ is an even plateau of height $h$. Then
  \[
  L_h(I) \leq \abs{I} + \sigma.
  \]
\end{lem}

\begin{cor}
  \label{cor:upper-load}
  Let $A$ be a local $(\rho, \sigma, \beta)$-bounded adversary with $\beta(i) \leq B$ for all $i$, and suppose $L$ is a configuration realized by OED forwarding. Suppose $I$ is an even plateau of height $h$. Then
  \[
  L_h(I) \leq \abs{I} + \sigma + \sum_{i \in I} \beta(i) \leq (B + 1) \abs{I} + \sigma.
  \]
\end{cor}
\begin{proof}
  For any interval $I$, let $A_I \subseteq A$ be the injection pattern consisting of packets injected into buffers $i \in I$. Since $A$ is locally $(\rho, \sigma, \beta)$-bounded, $A_I$ is $(\rho, \sigma, \beta')$-bounded with
  \[
  \beta'(i) =
  \begin{cases}
    \beta(i) &\text{if } i \in I\\
    0 &\text{otherwise}
  \end{cases}
  \]
  Therefore, by Observation~\ref{obs:global-vs-local-burst}, $A_I$ is $(\rho, \sigma')$-bounded, where $\sigma' = \sigma + \sum_{i \in I} \beta(i)$.

  Now let $I$ be an even plateau of height $h$. By Corollary~\ref{cor:even-plateau}, the configuration $L_h$ (see Equation~(\ref{eqn:h-load})) consists entirely of packets injected into $I$---i.e., packets injected by $A_I$. Since $A_I$ is a $(\rho, \sigma')$-bounded adversary, Lemma~\ref{lem:upper-load} implies that $L_h(I) \leq \abs{I} + \sigma'$, which gives the desired result.
\end{proof}

We now have all the pieces together to prove Theorem~\ref{thm:oed-ub}. The idea is to use Corollary~\ref{cor:upper-load} inductively to show that plateaus cannot grow too tall.

\begin{proof}[Proof of Theorem~\ref{thm:oed-ub}]
  Assume without loss of generality that $B$ is even, and consider any configuration $L$ attained by $A$. Let $i^* = \arg\max_i L(i)$ be a buffer with maximum load, and define $I_0 \supseteq I_1 \supseteq I_2 \supseteq \cdots \supseteq I_\ell$ where $I_j$ is the plateau of height $j$ containing $i^*$, and $\ell = L(i^*)$.

  Define $m$ to be the maximum value such that
  \begin{align}
    \label{eqn:m-def}
    \abs{I_{m(B+2)}} - \frac{\sigma}{B+1} \geq 1,
  \end{align}
  if such a value of $m > 0$ exists, and take $m = 0$ otherwise. Observe that for all $k$, we have
  \begin{align}
    \label{eqn:k-load}
    L_k(I_k) \geq \sum_{j = k+1}^\ell \abs{I_j}.
  \end{align}
  Since the $I_j$ are nested intervals, for any $k < m$ we have
  \begin{align}
    \label{eqn:k-load-b}
    L_{k(B+2)}(I_{k(B+2)}) \geq (B + 2) \sum_{j = k+1}^m L_{j (B + 2)}(I_{j (B+2)}).
  \end{align}
  Combining~(\ref{eqn:k-load-b}) with the result of Corollary~\ref{cor:upper-load}, we find that for all $j < m$,
  \begin{align}
    \label{eqn:upper-load-bound}
    (B+2) \sum_{j = k+1}^m \abs{I_{j(B+2)}} \leq (B + 1) \abs{I_{k(B+2)}} + \sigma.
  \end{align}
  Rearranging~(\ref{eqn:upper-load-bound}) yields that for all $k = 0, 1, \ldots m - 1$ we have
  \begin{align}
    \label{eqn:inductive-step}
    b \sum_{j = k+1}^m \abs{I_{j(B+2)}} - \frac{\sigma}{B+1} &\leq \abs{I_{k(B+2)}}\quad\text{where}\ b = \frac{B+2}{B+1}.
  \end{align}
  Note that for $k = m - 1, m - 2, \ldots$,~(\ref{eqn:inductive-step}) gives
  \begin{align}
    b \abs{I_{m(B+2)}} - \frac{\sigma}{B+1} &\leq \abs{I_{(m-1)(B+2)}} \label{eqn:ind-1}\\
    b \abs{I_{m(B+2)}} + b \abs{I_{(m-1)(B+2)}} - \frac{\sigma}{B+1} &\leq \abs{I_{(m-2)(B+2)}}\label{eqn:ind-2}\\
    &\vdots\nonumber
  \end{align}
  Combining~(\ref{eqn:ind-1}) and~(\ref{eqn:ind-2}), we obtain
  \begin{align}
    (b + b^2) \abs{I_{m(B+2)}} - (1 + b) \frac{\sigma}{B+1} \leq \abs{I_{(m-2)(B+2)}}.
  \end{align}
  Continuing in this way, a straightforward induction argument combined with the observation that $b > 1$ gives
  \begin{align}
    \label{eqn:m-exp-bound}
    b^m \paren{\abs{I_{m(B+2)}} - \frac{\sigma}{B+1}} \leq \abs{I_0} = n.
  \end{align}
  By the choice of $m$ in~(\ref{eqn:m-def}),~(\ref{eqn:m-exp-bound}) implies that
  \begin{align}
    m \leq \log_b n.
  \end{align}
  Again, from the definition of $m$, taking $h = m (B+2) + 2$, we have
  \[
  \abs{I_h} < 1 + \frac{\sigma}{B + 1}. 
  \]
  Applying Corollary~\ref{cor:upper-load}, gives
  \[
  L_h(I_h) \leq (B + 1) \abs{I_h} + \sigma \leq B + 2 \sigma + 1.
  \]
  Since there are at most $B + 2 \sigma + 1$ above $I_h$, the load of $i^*$ satisfies
  \begin{align*}
    L(i^*) &\leq h + B + 2 \sigma + 1\\
    &\leq (B + 2) m + B + 2\sigma + 3\\
    &= O(B \log n + \sigma),
  \end{align*}
  which gives the desired result.
\end{proof}

\subsection{Consequences}
\label{sec:ub-consequences}

Here, we list some consequences of the upper bound of Theorem~\ref{thm:oed-ub} when applied in combination with the packet bundling procedures described in Sections~\ref{sec:flows} and~\ref{sec:bundling}. For the following results, we assume that $G$ is an information gathering network with uniform edge capacity $C$.
\begin{cor}
  \label{cor:general-capacity}
  Suppose $A$ is a locally $(\rho, \sigma, \beta)$-bounded adversary with $\rho \leq C$ and $\beta(e) \leq B$ for all $e \in E$. Then OED forwarding applied to the $C$-reduction of $A$ has buffer space usage $O((B + C)\log n + \sigma)$.
\end{cor}

\begin{cor}
  \label{cor:bundling-continuous}
  Suppose $\Phi$ is a family of flows with locally dependent rate bound $r \leq C$, local burst parameters $b_\phi \leq B$, and global burst parameter $\sigma$. Then OED forwarding applied to the discretization of the $C$-reduction of $\Phi$ requires buffer space $O((B + C) \log n + \sigma)$.
\end{cor}

\begin{cor}
  \label{cor:bundling-heterogeneous}
  Suppose $A$ is a locally $(\rho, \sigma, \beta)$-bounded adversary with indivisible heterogeneous packet injections, and $\rho \leq C / 2$. Then OED forwarding applied to the bundled injection pattern described in Section~\ref{sec:heterogeneous} achieves buffer space usage $O((B + C) \log n + \sigma)$.
\end{cor}

  \section{Lower Bounds on Buffer Size}
\label{sec:buffer-size-lb}

In this section, we show that buffer space usage of the OED algorithm is asymptotically optimal among deterministic forwarding protocols. In Appendix~\ref{sec:rand-lb}, we generalize the lower bound to randomized forwarding protocols.

\begin{thm}
  \label{thm:det-lb}
  Let $F$ be any deterministic online forwarding protocol, and let $G = (V, E)$ be a single-destination path of length $n$. For any $B$ let $\beta(e) = B$ for all $e \in E$. Then there exists a local $(1, \sigma, \beta)$-bounded adversary $A$ such that
  \begin{equation}
    \label{eqn:lb}
    \sup_{t, i} L^t(i) = \Omega(B \log_B n + \sigma).    
  \end{equation}
\end{thm}

\begin{rem}
  The lower bound of~(\ref{eqn:lb}) is in contrast to the centralized and semi-local upper bounds of~\cite{Miller2016, Patt-Shamir2019-space-optimal, Miller2019-great}, which show that for $(\rho, \sigma)$-bounded adversaries, maximum buffer space $O(\rho + \sigma)$ (with no $n$ dependence) is achievable. In particular,~\cite{Patt-Shamir2019-space-optimal} gives a smooth tradeoff between the locality of a forwarding protocol and the optimal buffer space usage. Their work shows that if each node acts based on the state of its $d$-distance neighborhood, then $\Theta(\frac{1}{d}\log n + \sigma)$ buffer space is necessary and sufficient.\footnote{In the case $d = 1$, the algorithm of~\cite{Patt-Shamir2019-space-optimal} reduces to the OED algorithm.} Thus, for $(\rho, \sigma)$-bounded adversaries, the worst-case buffer space usage for a protocol generally depends on the protocol's locality ($d$). Together, Theorems~\ref{thm:oed-ub} and~\ref{thm:det-lb} show that this is not the case for local $(\rho, \sigma, \beta)$-bounded adversaries, as OED---a local ($d = O(1)$) protocol---is asymptotically optimal, even compared to centralized protocols. Thus, unlike for $(\rho, \sigma)$-bounded injection patterns, non-local information does not asymptotically improve the performance against local $(\rho, \sigma, \beta)$-bounded adversaries in information gathering networks.
\end{rem}

\begin{rem}
  In the proof of Theorem~\ref{thm:det-lb}, we describe an injection pattern as defined by an \emph{adaptive offline adversary}. That is, the choices made by the adversary are made in response to an algorithm's forwarding decisions (i.e., the current state of all buffers in the network). If the forwarding protocol is deterministic, this assumption about the adversary is without loss of generality, as the adversary can simulate the forwarding protocol and construct an injection pattern in advance. In Appendix~\ref{sec:rand-lb}, we will describe a (randomized) \emph{oblivious} adversary that is unaware of the forwarding protocol being used. Nonetheless, the oblivious adversary will almost surely require buffer space usage of $\Omega(B \log n + \sigma)$ against any (centralized, randomized) online forward protocol.
\end{rem}

\begin{proof}[Proof of Theorem~\ref{thm:det-lb}]
  Without loss of generality, we assume that the network size is a power of~$2B$, say, $n = (2B)^m$. Given a deterministic, online forwarding protocol $F$, we construct an adversary $A_F$ that injects packets in phases. Specifically, $A_F$ chooses a nested sequence of sub-intervals $I_0 \supseteq I_1 \supseteq I_2 \supseteq \cdots$ and in the $k\th$ phase, $A_F$ injects packets only into $I_k$. Each $I_k$ is chosen at the end of the $(k-1)\st$ phase depending on the loads of buffers in $I_{k-1}$. The $k\th$ phase lasts $\tau_k$ rounds.

  We set $\tau_1 = \frac 1 2 n$ and $I_1 = [n]$. Inductively, we define $\tau_k = \frac{1}{2B} \tau_{k-1}$ and $\abs{I_k} = \frac{1}{2B} \abs{I_{k-1}}$. At the beginning of phase $k$, $A_F$ selects $I_k$ and injects $B$ packets into each buffer in $I_k$. For $k > 1$, $A_F$ selects $I_k$ as follows. Let
  \[
  I_{k-1} = I_{k-1}^1 \cup I_{k-1}^2 \cup \cdots \cup I_{k-1}^{2B},
  \]
  where the $I_{k-1}^j$ are consecutive intervals of size $\frac{1}{2B} \abs{I_{k-1}}$. Then $A_F$ selects $I_k = I_{k-1}^j$ where $I_{k-1}^j$ has the largest total load at the end of the $k\th$ phase. That is, $j = \arg\max_j L(I_{k-1}^j)$. Observe that for all $k$ we have
  \begin{align}
    \tau_k &= \frac 1 2 \abs{I_k}, \text{ and}\\
    \tau_k &= B \abs{I_{k+1}}. \label{eqn:tau-i}
  \end{align}

  \begin{description}
  \item[Claim.] For all $k$, at the end of the $k\th$ phase, the total load of $I_k$ satisfies
    \begin{equation}
      \label{eqn:inductive-load}
      L(I_k) \geq k \paren{B - \frac 12} \abs{I_k}.
    \end{equation}
  \item[Proof of Claim.] We argue by induction on $k$. For $k = 1$, at the beginning of the first phase, $B$ packets are injected into each buffer in the network. Thus, the total load is $B n = B \abs{I_1}$. After $\abs{\tau_1} = \frac 1 2 n$ forwarding rounds, at most $\frac 1 2 n = \frac 1 2 \abs{I_1}$ packets are forwarded by the last buffer in $I_1$, hence the total load in $I_1$ is at least $(B - 1/2) \abs{I_1}$.

    For the inductive step, assume that $L(I_{k-1}) \geq (k - 1)(B - 1/2) \abs{I_{k-1}}$ at the end of the $(k-1)\st$ phase. By the choice of $I_k$, we therefore have $L(I_k) \geq (k - 1)(B - 1/2) \abs{I_k}$ at the end of the $(k-1)\st$ phase. At the beginning of the $k\th$ phase, $A_F$ injects $B \abs{I_k}$ packets into the buffers in $I_k$, hence the load becomes at least $(k B - (k-1) / 2) \abs{I_k}$. During the $\tau_k = \abs{I_k} / 2$ rounds of the $k\th$ phase, the last buffer in $I_k$ forwards at most $\tau_k = \abs{I_k} / 2$ packets, hence the total load of $I_k$ decreases by at most this amount. Thus we have $k (B - 1/2) \abs{I_k}$ as desired.
  \end{description}

  Applying the claim, after $k = \log_{2B} n$ phases, we have $\abs{I_k} = 1$ and $L(I_k) = k (B - 1/2) = \Omega(B \log_B n))$. The desired result follows by adding one more round in which $A_F$ injects $\sigma$ packets into $I_k$.

  All that remains is to show that $A_F$ is local $(1, \sigma, \beta)$-bounded. To this end, suppose each $k\th$ phase begins in round $s_k$ and ends in round $t_k$. Then injections only occur in rounds $s_k$. Now fix any subset $S$ of nodes and interval $T = [s, t]$, and define $j$ and $\ell$ such that $s_{j-1} < s \leq t_{j}$, and $\ell$ with $s_\ell \leq t \leq t_\ell$. For $k$ satisfying $j \leq k \leq \ell$, define $S_k = S \cap I_k$. Then observe that in round $s_k$, $B \abs{S_k}$ packets are injected into $S$. Therefore, the total number of packets injected into $S$ during $T$ is $N = \sum_{k = j}^\ell B \abs{S_k} \leq \sum_{k = j}^\ell B \abs{I_k}$, while the total number of rounds is $\abs{T} \geq \max\set{1, \sum_{k = j+1}^{\ell - 1} \tau_k}$. We bound $N$ as follows:
  \begin{align*}
    N &= B \abs{S_{j}} + B \abs{S_{j+2}} + \cdots + B \abs{S_\ell}\\
    &\leq B \abs{S_{j}} + B \abs{I_{j+2}} + \cdots + B \abs{I_\ell}\\
    &= B \abs{S_{j}} + \tau_{j+1} + \cdots + \tau_{\ell-1}\\
    &\leq B \abs{S} + \abs{T}.
  \end{align*}
  The second equality comes from Equation~(\ref{eqn:tau-i}). Thus, $A_{F}$ is local $(1, \sigma, \beta)$ bounded.
\end{proof}

In Appendix~\ref{sec:rand-lb}, we generalize the lower bound to randomized protocols.

  \bibliography{local-burst}

  \appendix
  
  \section{Generalization to Randomized Protocols}
\label{sec:rand-lb}

Here, we generalize the lower bound of Theorem~\ref{thm:det-lb} to \emph{randomized} forwarding protocols. Specifically, we construct a randomized \emph{oblivious} injection pattern that requires buffer space $\Omega(B \log n + \sigma)$ against any online forwarding protocol.

\begin{thm}
  \label{thm:rand-lb}
  Let $G = (V, E)$ be the single destination path of length $n$. Then for every $\sigma$, and $B$, there exists a randomized injection pattern $\Arand \in \calL(1, \sigma, B)$ such that for any (centralized, randomized) online forwarding protocol $F$, we have
  \[
  \sup_{t, i} L^t(i) = \Omega(B \log_B n + \sigma)
  \]
  almost surely.
\end{thm}

Again, we emphasize the order of quantifiers in the statement of the Theorem~\ref{thm:rand-lb}: the \emph{same} randomized injection pattern achieves the lower bound (almost surely) for every forwarding protocol.

The adversary $\Arand$ of Theorem~\ref{thm:rand-lb} is a straightforward modification of the adversary $A_F$ constructed the proof of Theorem~\ref{thm:det-lb}. Recall that $A_F$ injects packets into a nested sequence of intervals $I_0 \supseteq I_1 \supseteq \cdots \supseteq I_k$ where $k = O(\log_{B} n)$, and $\abs{I_k} = 1$. Each for $j \geq 1$, $I_j$ is chosen to be one of $2B$ sub-intervals of $I_{j-1}$ with maximum average load and $\abs{I_{j-1}} = 2 B \abs{I_j}$. The idea of $\Arand$ is to perform the same injection pattern as $A$, except that $\Arand$ chooses each $I_j \subseteq I_{j-1}$ randomly, independent of the choices of the forwarding protocol. We will show that for any execution of any forwarding protocol, injecting in this way yields a load of $\Omega(B \log n + \sigma)$ with probability $\Omega(1 / n)$. Thus, by independently repeating the randomized injection pattern ad infinitum, the lower bound is achieved almost surely (and with high probability after $O(n^2)$ injection rounds).

In order to formalize our description of $\Arand$, we first observe that the sequence $I_0 \supseteq I_1 \supseteq \cdots \supseteq I_k$ of intervals chosen by $A_F$ is uniquely determined by $I_k = [a_k]$, the final buffer into which $A_F$ injects packets. We also note that the injection pattern $A_F$ consists of $O(n)$ injection rounds, in which $O(B n + \sigma)$ packets in total are injected. Let $A_i$ denote such an injection pattern in which $I_k = [i]$---i.e., $i$ is the final buffer into which $A_i$ injects packets.

\begin{dfn}
  \label{dfn:a-rand}
  Let $A_i$ be the injection pattern described in the preceding paragraph. Then the adversary $\Arand$ injects packets as follows. Repeat:
  \begin{enumerate}
  \item choose $i \in [n]$ uniformly at random.
  \item in $O(n)$ rounds, inject packets as in $A_i$
  \item wait $B n + \sigma$ rounds without injecting any packets
  \end{enumerate}
  A single iteration of step~1--3 is an \dft{epoch}, and each epoch consists of $k = O(\log n)$ \dft{phases} (corresponding to the sub-intervals $I_j$ chosen in $A_i$).
\end{dfn}

\begin{dfn}
  Consider an execution of some forwarding protocol $F$ with adversary $\Arand$. We say that phase $j > 0$ of some epoch of $\Arand$ is \dft{good} if at the beginning of the phase we have
  \[
  \frac{L(I_j)}{\abs{I_j}} \geq \frac{L(I_{j-1})}{\abs{I_{j-1}}}.
  \]
  We say that an epoch is \dft{good} if all of its phases are good.
\end{dfn}

The following corollary follows immediately from the proof of Theorem~\ref{thm:det-lb}.

\begin{cor}
  \label{cor:lower-bound}
  Suppose an execution of a protocol $F$ with adversary $\Arand$ experiences a good epoch with injection pattern $A_i$. Then in the epoch's final injection round, $L(i) = \Omega(B \log n + \sigma)$.
\end{cor}

By Corollary~\ref{cor:lower-bound}, all that remains to prove Theorem~\ref{thm:rand-lb} is to show that each epoch is good with sufficient probability.

\begin{lem}
  \label{lem:good-phase}
  Consider a single epoch of an execution of a protocol $F$ with adversary $\Arand$. Then each phase is good independently with probability at least $1 / 2B$. 
\end{lem}
\begin{proof}
  Consider the interval $I_{j-1}$ at the beginning of the $j\th$ phase. There are $2 B$ choices of the sub-interval $I_j$, and each is chosen with equal probability. By the pigeonhole principle, at least one choice $I_j$ is good. Since $A_i$ is chosen uniformly at random with $i \in [n]$, conditioned on $I_{j-1}$, the choice of $I_{j}$ is uniformly at random. Thus, the probability $I_j$ is a good choice is at least $1 / 2B$. Finally, we note that since $i$ is chosen uniformly at random, the choice of \emph{which sub-interval} of $I_{j-1}$ is chosen in phase $j$ is independent of the sub-intervals chosen in other phases.
\end{proof}

\begin{cor}
  \label{cor:good-epoch}
  Each epoch is good independently with probability at least $1 / n$.
\end{cor}
\begin{proof}
  Each epoch consists of $\log_{2B} n$ phases, and each phase is good independently with probability $1 / 2B$. Therefore, the probability that all phases are good is $(1 / 2B)^{\log_{2B} n} = 1 / n$. By construction, the choices of $i$ used for each epoch are mutually independent.
\end{proof}

\begin{proof}[Proof of Theorem~\ref{thm:rand-lb}]
  The proof of Theorem~\ref{thm:det-lb} implies that each $A_i$ chosen by $\Arand$ is locally $(\rho, \sigma, \beta)$-bounded. Waiting $B n + \sigma$ rounds with no injections between epochs ensures that $\Arand$ is locally $(\rho, \sigma, \beta)$-bounded as well.

  By Corollary~\ref{cor:lower-bound}, if $\Arand$ experiences a good epoch, then it inflicts a buffer load of $\Omega(B \log n + \sigma)$. By Corollary~\ref{cor:good-epoch}, the probability that the first $k$ epochs are not good is at most $(1 - 1 / n)^k$. Taking the limit as $k \to \infty$, we find that $\Pr(\sup_{t, i'} L^t(i') = o(B \log n + \sigma)) = 0$.
\end{proof}

\begin{rem}
  The proof of Theorem~\ref{thm:rand-lb} shows that the probability that none of the first $k$ epochs are good is at most $(1 - 1 / n)^k$. This expression implies that a good epoch occurs with high probability after $k = O(n \log n)$ epochs. Since each epoch lasts $O(B n + \sigma)$ rounds, a good epoch (hence a load of $\Omega(B \log n + \sigma)$) occurs after $O(B n^2 \log n)$ rounds with high probability.
\end{rem}

\end{document}